\documentclass[11pt]{article}  
\usepackage{url}
\usepackage{graphicx}
\usepackage{amsfonts}
\usepackage{amssymb}
\usepackage{amsmath}
\usepackage{latexsym}
\usepackage{subfigure}
\usepackage{hyperref}

\newtheorem{theorem}{Theorem} 
\newtheorem{corollary}[theorem]{Corollary} 

\newcommand{\figlab}[1]{\label{fig:#1}}
\newcommand{\figref}[1]{\ref{fig:#1}}

\newcommand{\theolab}[1]{\label{theorem:#1}}
\newcommand{\theoref}[1]{\ref{theorem:#1}}

\newenvironment{proof}{\noindent\textbf{Proof: }\ignorespaces}
  {\hspace*{\fill}$\Box$\medskip}

\def\T{\mathbb{T}}

\newif\ifabstract
\abstracttrue
\newif\iffull
\ifabstract \fullfalse \else \fulltrue \fi

{\makeatletter
 \gdef\xxxmark{%
   \expandafter\ifx\csname @mpargs\endcsname\relax 
     \expandafter\ifx\csname @captype\endcsname\relax 
       \marginpar{xxx}
     \else
       xxx 
     \fi
   \else
     xxx 
   \fi}
 \gdef\xxx{\@ifnextchar[\xxx@lab\xxx@nolab}
 \long\gdef\xxx@lab[#1]#2{{\bf [\xxxmark #2 ---{\sc #1}]}}
 \long\gdef\xxx@nolab#1{{\bf [\xxxmark #1]}}
 \long\gdef\xxx@lab[#1]#2{}\long\gdef\xxx@nolab#1{}%
}

\title{A Universal Crease Pattern \\ for Folding Orthogonal Shapes}

\author{%
  \begin{tabular}{c@{\qquad\quad}c}
    Nadia M. Benbernou%
     \thanks{MIT Computer Science and Artificial Intelligence Laboratory,
       32 Vassar St., Cambridge, MA 02139, USA,
       \protect\url{{nbenbern,edemaine,mdemaine,avivo}@mit.edu}}
  &
    Erik D. Demaine\footnotemark[1]
      \thanks{Partially supported by NSF CAREER award CCF-0347776,
        DOE grant DE-FG02-04ER25647, and AFOSR grant FA9550-07-1-0538.}
  \\[\smallskipamount]
    Martin L. Demaine\footnotemark[1]
  &
    Aviv Ovadya\footnotemark[1]
  \end{tabular}
}
\date{}

\begin{document}
\maketitle

\begin{abstract}
  We present a universal crease pattern---known in geometry as 
  the \emph{tetrakis tiling} and in origami as \emph{box pleating}---that
  can fold into any object made up of unit cubes joined face-to-face
  (\emph{polycubes}).
  More precisely, there is one universal finite crease pattern
  for each number $n$ of unit cubes that need to be folded.
  This result contrasts previous universality results for origami,
  which require a different crease pattern for each target object,
  and confirms intuition in the origami community that box pleating is a
  powerful design technique.
\end{abstract}


\xxx{Gadgets + thms with slits}

\xxx{motions: does 1 gadget fold rigidly? rigid origami simulator? do slits make it easier...still global intersect issue}

\xxx{future: 0/1 matrix; general height map}

\section{Introduction}

An early result in computational origami is that every polyhedral surface
can be folded from a large enough square of paper
\cite{Demaine-Demaine-Mitchell-2000}.
But each such folding uses a different crease pattern.
Into how many different shapes can a single crease pattern fold?

Our motivation is developing programmable matter out of a foldable sheet.
The idea is to statically manufacture a sheet with specific creases,
and then dynamically program how much to fold each crease in the sheet.
Thus a single manufactured sheet can be programmed to fold into anything that
the single crease pattern (or a subset thereof) can fold.

We prove a universality result: a single $n \times n$ crease pattern
can fold into all face-to-face gluings of $O(n)$ cubes.
Thus, by setting the resolution $n$ sufficiently large,
we can fold any 3D solid up to a desired accuracy.


Our crease patterns are finite (rectangular) portions
of a single infinite tiling.
The \emph{tetrakis tiling} \cite{Gruenbaum-Shephard-1987}
is formed from the unit square grid by subdividing
each square in half vertically, horizontally, and by the two diagonals,
forming eight right isosceles triangles; see Figure~\figref{tetrakis}.
Note the scaling: one unit is the side length of an original square.
Equivalently, the tetrakis tiling can be formed from a half-unit square grid
with squares filled alternately with positive- and negative-slope diagonals.

\begin{figure}
  \centering
  \includegraphics[scale=0.6]{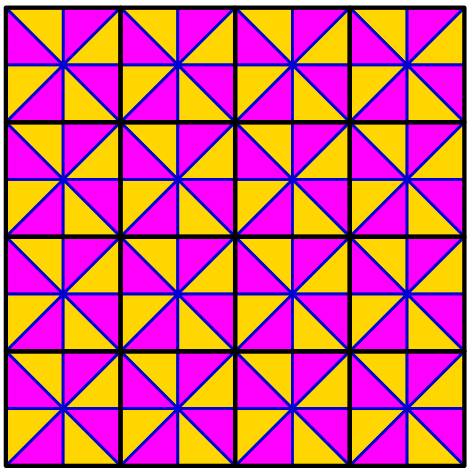}
  \caption{A $4 \times 4$ region of the tetrakis tiling.}
  \figlab{tetrakis}
\end{figure}

A \emph{tetrakis crease pattern} is a rectangular region (with integer
coordinates) of the tetrakis tiling.
Tetrakis crease patterns are similar to a style of origami
called \emph{box pleating} in which all creases are horizontal,
vertical, or diagonal, but their endpoints may not lie on the integer grid.
Box pleating was developed and popularized by Neal Elias in the 1960s
\cite{Kirschenbaum-Elias}, and explored more mathematically in
\cite{Lang-2003-secrets}.
Our universality result may have been suspected by origami artists,
but has not been proved until now.

\section{Definitions}

We start with a few definitions about origami,
specified somewhat informally for brevity.
For more formal definitions, see \cite[ch.~11]{Demaine-O'Rourke-2007}.

For our purposes, a \emph{piece of paper} is a two-dimensional surface
(formally, a metric 2-manifold).
A \emph{crease pattern} is graph drawn on the piece of paper
with straight edges and no crossings; edges are called \emph{creases}.
An \emph{angle assignment} is an assignment of real numbers in
$[-180^\circ,+180^\circ]$ to creases in the crease pattern,
specifying a fold angle (negative for valley, positive for mountain).
We allow a crease to be assigned an angle of~$0$,
in which case we call the crease \emph{trivial},
though we do not draw trivial creases in figures.
A crease pattern and angle assignment determine the 3D geometry of a folded
state, where each face maps as a 3D polygon via a rigid motion (isometry).
A \emph{folded state} consists of this geometry together with an
\emph{overlap order} defining the stacking relationship among faces of the
crease pattern that touch in the 3D geometry, allowing the paper to touch
but not cross itself.
We will specify such overlap orders visually using diagrams
that exaggerate the infinitesimal space between layers.

A \emph{polycube} is an interior-connected union of unit cubes from the
unit cube lattice.  The \emph{dual graph} of a polycube has a vertex
for each unit cube and an edge between two vertices whose corresponding
cubes share a face.
By the interior-connected property, the dual graph is connected.
The \emph{faces} of the polycube are the (square) faces of the individual cubes
that are not shared by any other cubes.

A \emph{folding of a polycube} is a folded state that covers all faces
of the polycube, and nothing outside the polycube.  In particular, we
allow the folded state to cover squares of the cubic lattice
interior to the polycube.  (Indeed, our foldings cover all such squares.)
A face of the folded polycube is \emph{seamless} if the outermost layer of
paper covering it is an uncreased unit square of paper.





\section{Folding Polycubes}

In this section, we describe an algorithm for folding a given $n$-cube
polycube from a square sheet of paper with crease pattern equal
to an $O(n) \times O(n)$ region of the tetrakis tiling.

%

\begin{theorem} \theolab{rectangle seam}
  Any polycube of $n$ cubes can be folded from a tetrakis crease pattern
  on a $(4 n + 1) \times (2 n + 1)$ rectangle of paper,
  with all faces seamless and made from one side of the paper,
  except for one specified face which has seams.
\end{theorem}

\begin{proof}
  The base case is $n=1$.
  Figures~\figref{OneCubeCrease} and \figref{Extrusion} show the crease
  pattern and folded state, respectively, for a single cube folded from a
  $5 \times 3$ rectangular sheet of paper, which is within the desired bound.
  For this base case, we fold just the shaded $5 \times 3$ part of the crease
  pattern in Figure~\figref{OneCubeCrease}, making exactly the desired cube.
  Although hidden in Figure~\figref{Extrusion}, the bottom face of the
  cube is indeed covered, with seams.  All other faces are seamless.

\begin{figure}
\centering
\includegraphics[width=\linewidth]{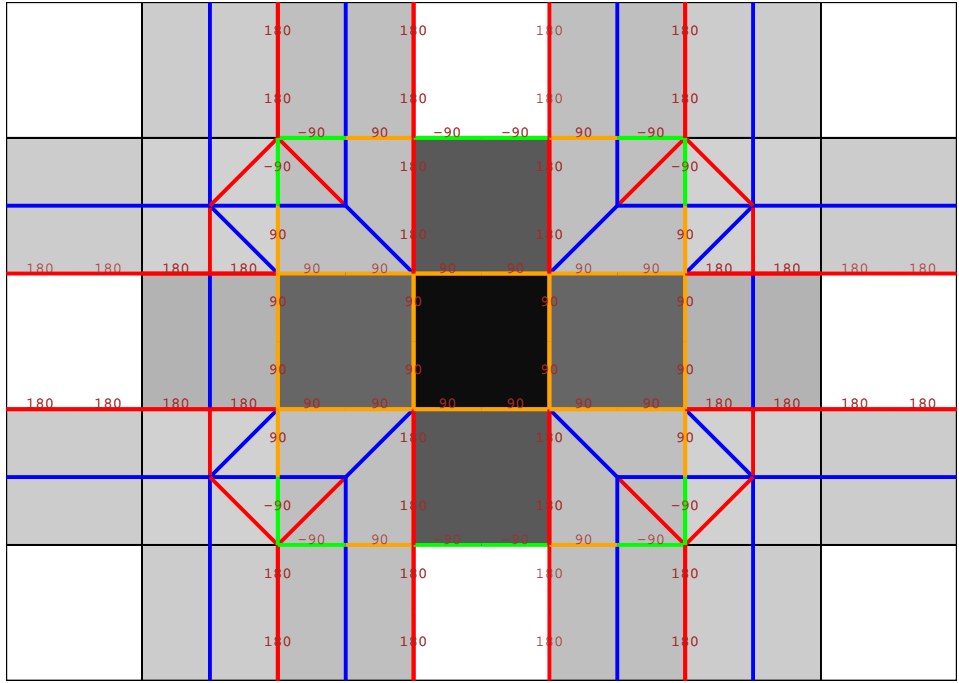}
\caption{Crease pattern for a folding a single unit cube.  Red line segments
  are mountain folds by $180^\circ$, orange segments are $90^\circ$ mountain
  folds, green segments are $90^\circ$ valley folds, and blue segments are
  $180^\circ$ valley folds.}
\figlab{OneCubeCrease}
\end{figure}

\begin{figure}
\centering
\includegraphics[width=\linewidth]{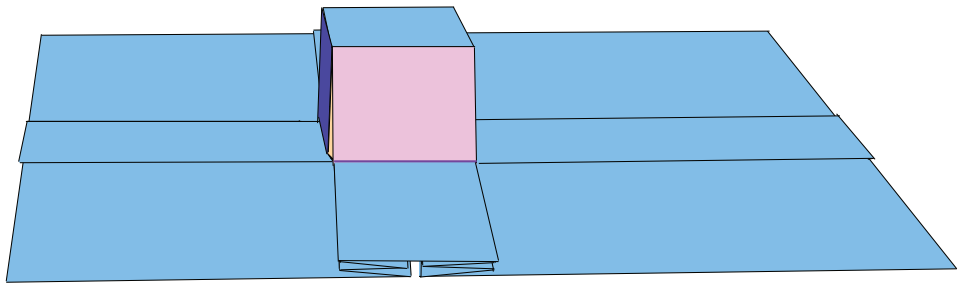}
\caption{Folding of a unit cube from the crease pattern
         in Figure~\protect\figref{OneCubeCrease}.}
\figlab{Extrusion}
\end{figure}

  It remains to prove the inductive step.

  Consider a polycube $P$ of $n$ cubes
  and a specified face $g \in P$ for seams.
  Let $b$ be the unique cube having $g$ as a face.
  Let $T$ be a spanning tree of the dual graph of~$P$,
  Because every tree has at least two leaves,
  $T$ has a leaf corresponding to a cube $l \neq b$.
  Let $t$ be the unique cube sharing a face with~$l$,
  and let $f$ be the face shared by $t$ and~$l$.

  Now consider the polycube $P' = P \setminus \{l\}$, with $n-1$ cubes.
  Because $l \neq b$, $g$ remains a face of~$P'$.
  By induction, the $(4(n-1)+1) \times (2(n-1)+1)$ tetrakis crease pattern
  folds $C'$ into $P'$ with an angle assignment~$A'$,
  without seams on all faces but~$g$.
  By symmetry we can assume that all faces are made
  from the top side of the paper.

  We modify $(C',A')$ as follows to obtain an angle assignment $A$ for the
  $(4 n + 1) \times (2 n + 1)$ tetrakis crease pattern $C$ folding into~$P$,
  without seams on all faces but~$g$.
  Because $f$ is a face of $P'$ and $f \neq g$,
  the folding of $(C',A')$ has $f$ seamless.
  Hence there is a unique unit square $s$ of $C'$
  corresponding to the outermost layer of paper covering $f$.
  Suppose that $s$ lies in row $i$ and column $j$ of~$C'$.
  We insert two columns to the left of column~$j$, two columns to the right of
  column~$j$, a single row above row~$i$, and a single row below row~$i$,
  and add the crease pattern shown in Figure~\figref{InsertionStep}.
  In particular, we have replaced $s$ by the $5 \times 3$ crease pattern for a
  unit cube, and we filled the remainder of the added rows and columns
  with the creases shown in Figure~\figref{OneCubeCrease},
  We also reflect the creases in the original row $i$ into the added rows,
  and similarly for the creases in the original column $j$, as described below.
  The new crease pattern $C'$ has width $(4(n-1)+1)+4 = 4 n + 1$
  and height $(2(n-1)+1)+2 = 2 n + 1$ as desired.

\begin{figure}[htbp]
\centering
\includegraphics[width=\linewidth]{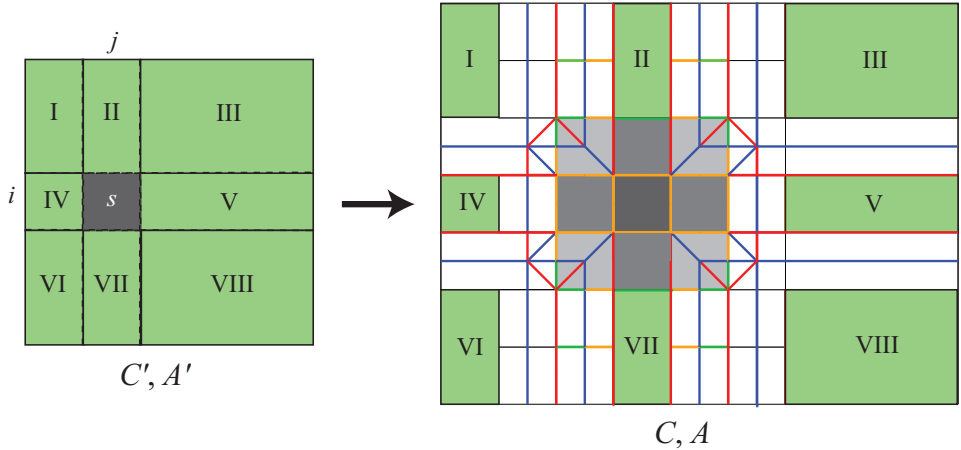}
\caption{Modifying $(C',A')$ (green) to produce $(C,A)$.
  Crease coloring is the same as Figure~\protect\figref{OneCubeCrease}.}
\figlab{InsertionStep}
\end{figure}

  Finally we show that the constructed crease pattern $C$ and
  angle assignment $A$ fold into the desired polycube~$P$.
  We construct the folded state in two steps.
  First, we fold using just the inserted crease pattern,
  producing a sheet with a cube $l$ sticking out in place of~$s$,
  as in Figure~\figref{Extrusion}.
  As mentioned in the base case, all faces of $l$ except $f$
  are seamless in this folding.
  Furthermore, all other unit squares of the sheet
  are seamless on the top side.
  Second, we apply the folding of $(C',A')$ to this folded object,
  pretending that the cube $l$ was just the unit square~$s$.
  Because $(C',A')$ does not fold~$s$, the folding of $(C',P')$ still works.
  The only difference is that all folds in row $i$ and column $j$
  apply now to three layers, not just one, causing additional creases
  in the inserted rows and columns.
  Because all faces of $P'$ are made from the top side of the paper,
  the folding remains seamless on all faces of $P'$ except $f$ and~$g$.
  Face $f$ is now the cube~$l$, which means that we have folded~$P$.
\end{proof}

A simple extension makes the folding entirely seamless:

\begin{corollary} \theolab{rectangle seamless}
  Any polycube of $n$ cubes can be folded from a tetrakis crease pattern
  on a $(4 n + 1) \times (2 n + 2)$ rectangle of paper,
  with all faces seamless and made from one side of the paper.
\end{corollary}

\begin{proof}
  Let $P$ be a polycube of $n$ cubes and let $f$ be any of its faces.
  We compute the folding from Theorem~\theoref{rectangle seam}
  of a $(4 n + 1) \times (2 n + 1)$ rectangle into~$P$,
  seamless except for~$f$.
  We add an extra column on the right,
  extending any nontrivial horizontal creases into this column.
  The resulting crease pattern and angle assignment fold into
  the desired polycube, with an extra seamless square
  attached along an edge of~$f$.
  We fold the seamless square on top of $f$
  to obtain an entirely seamless folding.
\end{proof}

Finally we show that a slightly more careful construction improves the
size of the required square of paper.

\begin{theorem} \theolab{cube}
  Any polycube of $n$ cubes can be folded from a tetrakis crease pattern
  on a square of paper of side length $3 n + 2$,
  with all faces seamless and made from one side of the paper.
\end{theorem}

\begin{proof}
  We follow the same construction as Theorem~\theoref{rectangle seam},
  but through the induction on~$n$, we alternate between the same modification
  and the $90^\circ$ rotation of the modification.  In other words,
  in odd steps we add four rows and two columns, and in even steps
  we add two rows and four columns.  Thus we add three rows and columns
  on average per step, starting from a $5 \times 3$ rectangle.
  For $n$ odd, we have an additional row, which is accounted for by the $+2$
  (instead of $+1$).
  In all cases, we have an additional column of paper, and for $n$ even,
  we have an additional row as well.  Folding these over in sequence,
  similar to Corollary~\theoref{rectangle seamless}, removes the seams
  from the last face.
\end{proof}


Note that these bounds are tight up to constant factors for square paper,
as folding an $n\times 1\times 1$ tower of unit cubes requires starting from
a square of side length $\Omega(n)$ in order to have diameter $\Omega(n)$,
because the diameter of the tower is $n$ and folding can only decrease
diameter.


\iffull

\section{Folding Arbitrary Orthogonal Shapes}

Let $\L^3$ denote a partitioning of the Euclidean 3-space by orthogonal planes. Let $R$ be a collection of rectangular boxes in $\L^3$ (where one rectangular prism corresponds to a volume completely enclosed by partitioning planes). 
Define the \emph{dual graph of $R$} as follows:  
There is one node per prism, and two nodes are joined via an edge if their corresponding prisms in $R$ share a face. We say that $R$ is \emph{connected} if its dual graph is connected. A \emph{polyprism} refers to a connected collection of a prisms.  

Let $P$ be a polyprism.  A face $f$ in $P$ is \emph{exposed} if it is the face of exactly one prism in $P$; in other words, no pair of prisms in $P$ shares $f$.  Let $C(P)$ be a crease pattern in $\T$ which folds to $P$.  A face in the folding of crease pattern $C(P)$ is \emph{seamless} if the outermost rectangle in the folding of that face has no creases. For any polyprism $P$, we will find a crease pattern $C(P)$ in $\T$ such that $C(P)$ folds to $P$. Only one exposed face of $P$ will not be seamless in the folding of $C(P)$, and we call this face the \emph{base face}. 

\begin{theorem}
Any polyprism \xxx{do this} cube $P$ of size $n$ can be folded from a square sheet of paper with side length $(3n+1)$ for even $n$ and $(3n+2)$ for odd $n$.  Moreover the crease pattern used is a subgraph of the tetrakis tiling.
\theolab{PolycubeResult}
\end{theorem}

\section{Variations}
\xxx{complete draft: this is more of notes of what to write later than anything else. it would be much more elegant if we first introduced the 22.5 degree version, and then modify the bounds for box pleating.}
\subsection{Folding Polycubes when allowing 22.5 degree angles}
If the one allows additional crease angles such as those shown in Figure \xxx{add figure}, then additional paper is not needed to create `buffer' regions. This decreases both the total number of needed tiles and the total number of layers.
\subsection{Folding Polycubes Using Paper with Slits}
Similarly, if the folded paper has a pattern of slits such as that shown in Figure \xxx{add figure}, then additional paper is not needed to create the `buffer' regions, having the same effect as 22.5 degree angles on all bounds. However, this introduces many `seams' into the paper, where slits lie along faces. Ensuring that all faces are seamless requires a custom slit pattern for a particular polycube, or a different general pattern with additional paper and layers. An algorithm for coming up with an optimal general slit pattern given a particular set of polycube to be folded is an open problem.

\fi

\bibliography{proof}
\bibliographystyle{alpha}

\end{document}